\documentclass[a4paper, 11pt]{article}

\usepackage{amsmath, amssymb,amsthm,a4wide}       

\usepackage[colorlinks=true]{hyperref}

\newtheorem{lem}{Lemma}
\newtheorem{thm}{Theorem}
\newtheorem{cor}{Corollary}

\newcommand{\DD}{\mathcal{D}}

\newcommand{\VV}{\mathcal{V}}

\newcommand{\PP}[1]{{\text{\normalsize$\mathbb P$}}\left(#1\right)}
\newcommand{\tr}[1]{\text{Tr}\left(#1\right)}
\newcommand{\trr}[1]{\text{Tr}^2\left(#1\right)}

\newcommand{\ML}{\text{\tiny$\!M\!L$}}
\newcommand{\BM}{\text{\tiny$\!B\!M$}}

\newcommand{\Nm}{N}

\newcommand{\brho}{\overline{\rho}}
\newcommand{\blambda}{\overline{\lambda}}
\newcommand{\bP}{\overline{P}}
\newcommand{\bF}{\overline{\boldsymbol{F}}}

\newcommand{\dv}[2]{\frac{\partial #1}{\partial #2}}
\newcommand{\dvv}[2]{\frac{\partial^{2} #1}{\partial #2 ^{2}}}

\begin{document}
\title{Asymptotic expansions of  Laplace integrals for  quantum state tomography }

\author{Pierre Six and Pierre Rouchon\thanks{Centre Automatique et Syst\`{e}mes, Mines-ParisTech, PSL Research University. 60 Bd Saint-Michel,
75006 Paris, France.}}

\maketitle

\abstract{
Bayesian estimation  of a mixed quantum state  can be approximated via maximum likelihood (MaxLike)  estimation when the likelihood function is sharp around its maximum. Such approximations rely  on asymptotic expansions of multi-dimensional  Laplace integrals.  When this maximum is on the boundary of the integration domain, as it is the case when the MaxLike  quantum state is not full rank, such expansions are not standard. We provide here such expansions, even  when  this maximum does not belong to the smooth part of the boundary, as it is the case when the rank deficiency  exceeds  two.  These expansions provide, aside the MaxLike estimate of the quantum state, confidence intervals for any observable. They confirm  the formula  proposed and used  without  precise  mathematical justifications  by the authors  in an article recently published in Physical Review A.  }

\section{Introduction}
	When the probability laws of  the measurement data  $Y$ with respect to the continuous parameter $p$ to estimate is given by an analytic  model, a  widely used  way to fulfil this estimation is Maximum Likelihood (MaxLike) reconstruction (see, e.g., \cite{MoulinesBook2005}). It consists in choosing as estimate $p_{\ML}$,  the value of $p$ that maximizes the conditional probability $\PP{Y~|~p}$ of the data  $Y$. Indeed, when the amount of independent  measurements  forming the  data $Y$ is  large, the function $p \mapsto \PP{Y~|~p}$ becomes extremely sharp around its maximal value, and the MaxLike estimate $p_{\ML}$ is a good approximation of the Bayesian mean estimate  $p_{\BM}$:
 $$
	p_{\BM} = \int_{\DD} p \, \PP{p~|~Y}  \, \mathrm dp = \frac{\int_{\DD} p \, \PP{Y~|~p}  \, \mathbb{P}_0(p) \, \mathrm dp}
			{\int_{\DD} \PP{Y~|~p}  \, \mathbb{P}_0(p) \, \mathrm dp}
	$$
with $\DD\subset \mathbb{R}^{\dim p}$ being the set  of physically acceptable values for $p$, $\PP{p~|~Y}$ the probability density of $p$ knowing $Y$ and $\mathbb{P}_0(p)$ any a priori probability  density for $p$.
	
Relying only on MaxLike estimation has the advantage of providing easy-to-compute algorithms. The first and second derivatives of $\PP{Y~|~p} $ versus $p$ can be derived with finite difference method, gradient-like optimization methods can be used and one can extract the Cram\'er-Rao bound from the Hessian of the log-likelihood function to get a lower bound of the mean estimation error when this Hessian matrix is not degenerate. Nevertheless,  some technicalities can arise, in particular  for quantum state tomography~\cite{paris-rehacek-book2004}, where the parameter $p$ to estimate corresponds to  a quantum  state $\rho$ an element of  the compact convex domain  $\DD$ formed by the set of non negative Hermitian matrix of trace one.  In practice, MaxLike estimates $\rho_{\ML}$ could be  of low-rank, i.e. on the boundary   of $\DD$  as noticed  in~\cite{Blume2010NJoP} and observed in~\cite{SixPRA2016}.

All these reasons lead us to  consider   Bayesian Mean Estimations (BME) in the general setting when the parameter $p$ lives in a finite dimensional and compact  domain $\DD$ with piece-wise smooth boundary.  As the magnitude  of $\PP{Y~|~p}$ grows (or decreases) exponentially fast compared to the number  $ \Nm $ of independent measurements generating the measurement set $Y$, we consider the scaled  log-likelihood function  $f(p) = \tfrac{1}{ \Nm } \log\left(\PP{Y~|~p}\right)$. We then address  the problem of computing, for any smooth  scalar functions $f$ and $g$ and under various conditions, the asymptotic development when $ \Nm $ tends towards infinity of  the Laplace's integral:
	\begin{equation} \label{eq:Ig_def}
	\mathcal{I}_g( \Nm ) = \int_{\DD} g(p) \exp\left( \Nm  f(p)\right) \,  \mathrm dp.
	\end{equation}
Such asymptotic expansions have been investigated since a long time. They involve  integration by parts, Watson's lemma, Laplace's method, stationary phase, steepest descents  and Hironaka's resolution of singularities: see \cite{BleistenHandelsmanBook} for $\dim p=1$ and   the  regular  case  when $\dim p \geq 1$; see \cite{ArnoldGuseinZadeVarchenkoVol2} for the singular case  in arbitrary dimension    and its much  more elaborate analysis. In the analytic case  and around the maximum  of $f$ at $p_\ML$ inside the domain $\DD$, these expansions rely on terms like $e^{\Nm f(p_\ML)}\frac{(\log\Nm)^k}{\Nm^{\alpha}}$ where $k$ is a non negative  integer less than $\dim p -1$ and where $\alpha$ is rational  and strictly positive \cite[page 231]{ArnoldGuseinZadeVarchenkoVol2}. From such series expansions, stem fundamental connections between   algebraic geometry and statistical learning theory in the singular case, i.e. when the Hessian of $f$ at $p_\ML$ is not  negative definite. This is the object of singular learning theory developed in~\cite{WatanabeBook2009} and in~\cite{ShaoweiLinPhD2011}.

It is interesting to notice that, as far as we know,   very few results  can be found when $p_\ML$ lies on the boundary of $\DD$, excepted the case when $p_\ML$ is on a smooth part of the boundary.
In~\cite[section 8.3]{BleistenHandelsmanBook}, the derivation of   the leading term is explained  when $p_\ML$ is on the smooth part of the boundary and when the Hessian of the  restriction of $f$ to this smooth part is  negative definite; sub-section 8.3.4 of~\cite{ArnoldGuseinZadeVarchenkoVol2} provides precise indications showing, when the Hessian of the  restriction of $f$   is degenerate, that  an asymptotic expansion  exists and is  similar to the one obtained for  $p_\ML$  in the interior of $\DD$.

For quantum state estimation, this ensures the existence of asymptotic expansion in any case when $\rho_\ML$ has  either a  full rank (interior of $\DD$) or  rank deficiency of one (smooth part of the boundary of $\DD$). For rank deficiency exceeding strictly one,  $\rho_\ML$  does not belong to  the smooth part  of the boundary. As far as we know, the derivations of  asymptotic expansions  in these singular cases when the rank deficiency  of $\rho_\ML$ exceeds two     have not been precisely addressed up to now. This  paper  is a fist attempt to derive such asymptotic expansion of the Bayesian mean and variance when the log-likelihood function  reaches it maximum  on the boundary of $\DD$, i.e. when $\rho_\ML$ is of low rank.

The goal of this paper is  twofold. Firstly, we provide the leading terms of   specific  asymptotic  expansions when $p_\ML$ lies in an half space.    This is the object of section~\ref{sec:boundary} where we assume  that the restriction of $f$ to the boundary    admits a non-degenerate  maximum at $p_\ML$ (see theorem~\ref{thm:boundary}).  Secondly,  we consider quantum state estimation and  reformulate these leading terms intrinsically in terms of matrix product and trace.  This is object of section~\ref{sec:Qtomo}, where we recall  the precise structures of $f$ and $g$ in this case and exploit convexity  and unitary invariance.  We provide in this section precise mathematical justifications of  the  necessary and sufficient optimality conditions given without details in~\cite[eq. (8)]{SixPRA2016} (see lemma~\ref{lem:Qtomo1} below)  and  of the Bayesian variance  approximation corresponding to equation~(10) in~\cite{SixPRA2016} (see theorem~\ref{thm:Qtomo2}).

\section{Asymptotic expansion of Laplace's integral} \label{sec:boundary}

Here, we assume that $p$ is of dimension $n$  and that $\DD= (-1,1)^{n} $. Set $p=z$ with $z\in\mathbb{R}^{n}$.  Then~\eqref{eq:Ig_def} reads with:
\begin{equation} \label{eq:IgLocIn}
	\mathcal{I}_g( \Nm ) = \int_{z\in(-1,1)^{n}}  g(z) \exp\left( \Nm  f(z)\right) \,  \mathrm dz
.
\end{equation}
   \begin{thm}\label{thm:interior}
     Consider~\eqref{eq:IgLocIn} where $f$ and $g$ are analytic functions of $z$ on a compact neighbourhood of $\overline{\DD}$, the closure of $\DD$.  Assume that $f$ admits a unique maximum on $\overline{\DD}$ at $z=0$ with  $\left.\dvv{f}{z}\right|_{0}$ negative definite.

     If $g(0)\neq 0$,  we have the following  dominant term in the asymptotic expansion of $\mathcal{I}_g( \Nm )$ for large $\Nm$:
     \begin{equation} \label{case1:thm:interior}
        \mathcal{I}_g( \Nm ) = \left(\frac{g(0) ~(2\pi)^{n/2} ~e^{\Nm f(0)}
        \Nm^{-n/2}  }{\sqrt{\left| \det\left(\left.\dvv{f}{z}\right|_{0}\right) \right|~ }}\right)
           + O\Big(e^{\Nm f(0)}\Nm^{-n/2-1}\Big)
        .
     \end{equation}

     If $g(0)=0$, with $\left.\dv{g}{x}\right|_{0}=0$ and $\left.\dv{g}{z}\right|_{0}=0$, then we have:
         \begin{multline} \label{case2:thm:interior}
     \mathcal{I}_g( \Nm ) =
     \left(\frac{\tr{-\left.\dvv{g}{z}\right|_{0} \left(\left.\dvv{f}{z}\right|_{0}\right)^{-1} } ~(2\pi)^{n/2}}{2 ~\sqrt{\left| \det\left(\left.\dvv{f}{z}\right|_{0}\right) \right|~ } } \right)
     e^{\Nm f(0)}
        \Nm^{-n/2-1}
        \\
+ O\left(e^{\Nm f(0)} \Nm^{-n/2-2}\right)
     .
        \end{multline}
   \end{thm}

  \begin{proof}
Since $f$ is analytic,
$
f(z)= f(0) -   h(z)
$
 where $h$ is an analytic function of $z$ only with $h(0)=0$, $\left.\dv{h}{z}\right|_{0}=0$ and $\left.\dvv{h}{z}\right|_{0}=-\left.\dvv{f}{z}\right|_{0}$ positive  definite.

Via the Morse lemma (see, e.g., \cite{MilnorBook63}), there exists a local diffeomorphism on $z$  around  $0$, written $\tilde z=\psi(z)$, such that $\psi(0)=0$ and  $h(z)=  \frac{1}{2} \sum_{k=1}^{n} (\psi_k(z))^2$. Moreover, we can chose $\psi$  such that $\left.\dv{\psi}{z}\right|_{0}=\sqrt{-\left.\dvv{f}{z}\right|_{0}}$ is a positive definite symmetric matrix.

Take $\eta\in(0,1)$ small. There exists a $c < f(0)$ such that, $\forall z\in(-1,1)^n/(-\eta,\eta)^n$, $f(z)\leq c$. Since:
\begin{multline*}
 \mathcal{I}_g( \Nm ) = \int_{z\in(-\eta,\eta)^n}   g(z) e^{\Nm  f(z)} \, \mathrm dz
 +\int_{z\in(-1,1)^n/(-\eta,\eta)^n}  g(z) e^{\Nm  f(z)} \,  \mathrm dz
 \\
 =
  e^{\Nm  f(0)}
 \left( \int_{z\in(-\eta,\eta)^n} g(z) e^{\Nm  (f(z)-f(0))} \, \mathrm dz +
O\left( e^{-\Nm  (f(0)-c)}\right)
 \right),
\end{multline*}
we only keep:
$$
I_\eta(\Nm)= \int_{z\in(-\eta,\eta)^n}g(z) e^{\Nm  (f(z)-f(0))} \, \mathrm  dz
$$
Since $\eta$ is small,  we can consider the change of variable $\tilde z=\psi(z)$ that yields:
$$
  I_\eta(\Nm)= \int_{\tilde z\in\psi((-\eta,\eta)^n)}  \tilde g(\tilde z) e^{ - \frac{\Nm}{2}\sum_{k=1}^{n} \tilde z_k^2  } \, \mathrm   d\tilde z
$$
where:
\begin{equation}\label{eq:gtildeInt}
  \tilde g(\tilde z)=  \frac{ g(\psi^{-1}(\tilde z))  }{ \sqrt{\left| \det\left(\left.\dvv{f}{z}\right|_{0}\right) \right| } } ~(1+\tilde d(\tilde z))
\end{equation}
and $\tilde d$ is an analytic function with $\tilde d(0)=0$. There exists $\tilde \eta >0$ such that $ (-\tilde \eta,\tilde \eta)^n \subset \psi((-\eta,\eta)^n)$. Thus, similarly to the passage from
$ \mathcal{I}_g( \Nm ) $ to $I_\eta(\Nm)$, we can, up to exponentially small terms versus $\Nm$,  just  consider the asymptotic expansion of:
\begin{equation}
\label{eq:Itilde_z}
\tilde I_{\tilde \eta} = \int_{\tilde z \in (-\tilde \eta,\tilde \eta)^n}\tilde g(\tilde z) e^{ -\frac{\Nm}{2}\sum_{k=1}^{n} \tilde z_k^2  } \, \mathrm d\tilde z
.
\end{equation}

When $g(0)\neq 0$, we have  $\tilde g(0) \neq 0$. Set $\tilde g(\tilde z)= \tilde g(0) + \sum_{k=1}^n \tilde z_k \tilde h_k(\tilde z) $ with  $\tilde h_k$ bounded analytic functions  on $(-\tilde \eta,\tilde \eta)^n$. We get:
$$
 \tilde I_{\tilde \eta} = \tilde g(0) \int_{\tilde z \in (-\tilde \eta,\tilde \eta)^n }e^{ -\frac{\Nm}{2}\sum_{k=1}^{n} \tilde z_k^2  }~ \, \mathrm d\tilde z
 +\int_{\tilde z \in (-\tilde \eta,\tilde \eta)^n } \left( \sum_{k=1}^n \tilde z_k \tilde h_k(\tilde z) \right) e^{ - \frac{\Nm }{2}\sum_{k=1}^{n} \tilde z_k^2 } \, \mathrm d\tilde z
 .
$$
Up to exponentially small terms versus $\Nm$,  the first integral in the right hand-side member can be replaced by:
$$
 \int_{\tilde z \in  (-\infty,+\infty)^n} e^{ - \frac{\Nm}{2}\sum_{k=1}^{n} \tilde z_k^2 } \,  \mathrm d\tilde z
 =  \left(\frac{2\pi}{\Nm}\right)^{n/2}
 .
 $$
A single integration by part versus $z_k$ yields:
 \begin{multline*}
     \int_{\tilde z \in (-\tilde \eta,\tilde \eta)^n } \tilde z_k \tilde h_k(\tilde z) e^{ - \frac{\Nm}{2}\sum_{k=1}^{n} \tilde z_k^2 } \, \mathrm  d\tilde z
     \\
     = \frac{1}{\Nm}\int_{\tilde z \in (-\tilde \eta,\tilde \eta)^n }
      \dv{\tilde h_k}{\tilde z_k}(\tilde z)
     e^{ - \frac{\Nm}{2}\sum_{k=1}^{n} \tilde z_k^2 } \, \mathrm  d\tilde z
     + O(e^{-\tilde \eta^2 \Nm/2}/\Nm)
     .
 \end{multline*}
     This implies~\eqref{case1:thm:interior}, via
 $\tilde I_{\tilde \eta} =   \tilde g(0) \left(\frac{2\pi}{\Nm}\right)^{n/2} \left( 1 + O(1/\Nm)\right)$ and $\tilde g(0)=   \frac{ g(0)}{ \sqrt{\left| \det\left(\left.\dvv{f}{z}\right|_{0}\right) \right| } } $.

 Assume now that  $g(0)=0$  and $\left.\dv{g}{z}\right|_{0}=0$. Consider then the function  $\tilde g$ in ~\eqref{eq:gtildeInt}. We have  $\tilde g(0)=0$ and $\left.\dv{\tilde g}{\tilde z}\right|_{0}=0$. Moreover, writting:
 $$
 \kappa_0=\sqrt{\left| \det\left(\left.\dvv{f}{z}\right|_{0}\right) \right|~ },
 $$
  we have:
 $$\kappa_0 \tilde g(\psi(z)) = g(z) ~(1+ e(z))$$ with $e$ an analytic  function with  $e(0)=0$. Thus, for any $i,j\in\{1,\ldots, n\}$,
 $$
\left.\frac{\partial^2 g}{\partial z_i\partial z_j}\right|_{0}=  \kappa_0 \sum_{k,k'=1}^{n}\left.\frac{\partial ^2 \tilde g}{\partial \tilde z_k\partial \tilde z_{k'}}\right|_{0} \left.\dv{\psi_k}{z_i}\right|_{0}\left.\dv{\psi_{k'}}{z_i}\right|_{0}.
 $$
 Since $\left.\dv{\psi}{z}\right|_{0}= \sqrt{-\left.\dvv{f}{z}\right|_{0}}$, we have:
 $$
 \kappa_0 \left.\dvv{\tilde g}{\tilde z}\right|_{0} = \left(\sqrt{-\left.\dvv{f}{z}\right|_{0}}\right)^{-1} \left.\dvv{g}{z}\right|_{0} \left(\sqrt{-\left.\dvv{f}{z}\right|_{0}}\right)^{-1},
 $$
 and thus:
 \begin{equation}\label{eq:tracegInt}
 \tr{\left.\dvv{\tilde g}{\tilde z}\right|_{0}} = \frac{\tr{-\left.\dvv{g}{z}\right|_{0} \left(\left.\dvv{f}{z}\right|_{0}\right)^{-1} }}{\sqrt{\left| \det\left(\left.\dvv{f}{z}\right|_{0}\right) \right|}}
\end{equation}
Since $\tilde g$ and its first partial derivatives with respect to $\tilde z_k$ vanish, we have:
$$
\tilde g(\tilde z)=  \sum_{k,k'=1}^n \tilde z_k \tilde z_{k'} \tilde b_{k,k'} (\tilde z),
$$
where the function  $\tilde b_{k,k'}$  are analytic . To evaluate the integral in~\eqref{eq:Itilde_z}, we have to consider  the dominant terms of the following integrals:
$$
     B_{k,k'}= \int_{\tilde z \in (-\tilde \eta,\tilde \eta)^n } \tilde z_k \tilde z_{k'} \tilde b_{k,k'} (\tilde z)  e^{ - \frac{\Nm}{2}\sum_{l=1}^{n} \tilde z_l^2  } \, \mathrm d\tilde z
  .
$$
 For $k\neq k'$,  one integration by part versus $\tilde z_k$ followed by another one versus $\tilde z_{k'}$, yield    to  $B_{k,k'}= O\left(\Nm^{-n/2-2}\right)$. For $k=k'$,  we can perform a single integration by part versus $\tilde z_k$:
 \begin{multline*}
\int_{\tilde z \in (-\tilde \eta,\tilde \eta)^n } \tilde z_k^2 \tilde b_{k,k} (\tilde z)  e^{ - \frac{\Nm}{2}\sum_{l=1}^{n} \tilde z_l^2  } \, \mathrm d\tilde z
\\
= \frac{1}{\Nm} \int_{\tilde z \in (-\tilde \eta,\tilde \eta)^n } \left( \tilde b_{k,k} (\tilde z)+ \tilde z_k \dv{\tilde b_{k,k}}{\tilde z_k} (\tilde z)\right)
 e^{ - \frac{\Nm}{2}\sum_{l=1}^{n} \tilde z_l^2  } \, \mathrm d\tilde z + O(e^{-\Nm\tilde \eta^2/2})
\\
=
 \frac{\tilde b_{k,k}(0)}{\Nm} \left(\frac{2\pi}{\Nm}\right)^{n/2}
+ O\left(\Nm^{-n/2-2}\right)
.
 \end{multline*}
The sum $ \sum_{k,k'} B_{k,k'}$  corresponds to   the integral $ \tilde I_{\tilde \eta}$ and  reads:
$$
\tilde I_{\tilde \eta}(\Nm) = \frac{\sum_{k=1}^{n} \tilde b_{k,k}(0)}{\Nm} \left(\frac{2\pi}{\Nm}\right)^{n/2}
+ O\left(\Nm^{-n/2-2}\right)
.
$$
Since up to exponentially small terms,  $ \tilde I_{\tilde \eta}$ and $e^{-\Nm f(0)} \mathcal{I}_g( \Nm ) $ coincide, we  get~\eqref{case2:thm:interior} using~\eqref{eq:tracegInt}  since
$
\sum_{k=1}^{n} \tilde b_{k,k}(0)=\tfrac{1}{2} \tr{\left.\dvv{\tilde g}{\tilde z}\right|_{0}}
$. 
\end{proof}

We assume now that  $p\in\mathbb{R}^{n+1}$, $n+1$ being the dimension of $p$  ($n$ non-negative integers),  and that $\DD=(0, 1)\times (-1,1)^{n} $. Set $p=(x,z)$ with $x\in\mathbb{R}$ and $z\in\mathbb{R}^{n}$.  Then~\eqref{eq:Ig_def} reads when  $g(x,z)$ is replaced by $x^m g(x,z)$, with $m$ a non negative integer:
\begin{equation} \label{eq:IgLoc}
	\mathcal{I}_g( \Nm ) = \int_{x\in (0,1)} \int_{z\in(-1,1)^{n}} x^m g(x,z) \exp\left( \Nm  f(x,z)\right) \,  \mathrm dx~\mathrm dz
.
\end{equation}
   \begin{thm}\label{thm:boundary}
     Consider~\eqref{eq:IgLoc}, where $f$ and $g$ are analytic   functions of $(x,z)$ on a compact neighbourhood of $\overline{\DD}$, the closure of $\DD$.  Assume that $f$ admits a unique maximum on $\overline{\DD}$ at $(x,z)=(0,0)$, with  $\left.\dvv{f}{z}\right|_{(0,0)}$ negative definite and   $\left.\dv{f}{x}\right|_{(0,0)}  <0$.

     If $g(0,0)\neq 0$,  we have the following  dominant term in the asymptotic expansion of $\mathcal{I}_g( \Nm )$ for large $\Nm$:
     \begin{equation} \label{case1:thm:boundary}
        \mathcal{I}_g( \Nm ) = \left(\frac{g(0,0) ~m! ~(2\pi)^{n/2} ~e^{\Nm f(0,0)}
        \Nm^{-m-n/2-1}  }{\sqrt{\left| \det\left(\left.\dvv{f}{z}\right|_{(0,0)}\right) \right|~ } \left(-\left.\dv{f}{x}\right|_{(0,0)}\right)^{m+1}}\right)
           + O\Big(e^{\Nm f(0,0)}\Nm^{-m-n/2-2}\Big)
        .
     \end{equation}

     If $g(0,0)=0$, with $\left.\dv{g}{x}\right|_{(0,0)}=0$ and $\left.\dv{g}{z}\right|_{(0,0)}=0$, then we have:
         \begin{multline} \label{case2:thm:boundary}
     \mathcal{I}_g( \Nm ) =
     \left(\frac{\tr{-\left.\dvv{g}{z}\right|_{(0,0)} \left(\left.\dvv{f}{z}\right|_{(0,0)}\right)^{-1} }~m! ~(2\pi)^{n/2}}{2 ~\sqrt{\left| \det\left(\left.\dvv{f}{z}\right|_{(0,0)}\right) \right|~ } \left(-\left.\dv{f}{x}\right|_{(0,0)}\right)^{m+1}} \right)
     e^{\Nm f(0,0))}
        \Nm^{-m-n/2-2}
        \\
+ O\left(e^{\Nm f(0,0))} \Nm^{-m-n/2-3}\right)
     .
        \end{multline}
   \end{thm}
   For clarity's sake, we consider here  the analytic  situation, despite the fact that  the above asymptotics are  also valid in the $C^{m+3}$ case.

\begin{proof}
  We adapt  here the method sketched  in section 8.3.4 of ~\cite{ArnoldGuseinZadeVarchenkoVol2} for oscillatory integrals in a halfspace.
Since $f$ is analytic, we have
$$
f(x,z)= f(0,0) - x  f_1(x,z) -  h(z)
$$
where  $f_1$ is analytic with $f_1(0,0)=-\left.\dv{f}{x}\right|_{(0,0)}>0$, where $h$ is an analytic function of $z$ only, with $h(0)=0$, $\left.\dv{h}{z}\right|_{0}=0$ and $\left.\dvv{h}{z}\right|_{0}=-\left.\dvv{f}{z}\right|_{(0,0)}$ positive  definite.

Set $\phi(x,z)= x f_1(x,z) $. Consider the following map
$(x,z) \mapsto (\tilde x= \phi(x,z), z)$. It is a local diffeomorphism around $(0,0)$ that preserves the sign of $x$, i.e. $x \phi(x,z) \geq 0$.  Moreover, using the Morse lemma (see, e.g., \cite{MilnorBook63}), there exists a local diffeomorphism on $z$  around  $0$,  $\tilde z=\psi(z)$, such that $\psi(0)=0$ and  $h(z)=  \frac{1}{2} \sum_{k=1}^{n} (\psi_k(z))^2$ (see, e.g., \cite{MilnorBook63}). Moreover, we can chose $\psi$  such that $\left.\dv{\psi}{z}\right|_{0}=\sqrt{-\left.\dvv{f}{z}\right|_{0}}$ is a positive definite symmetric matrix.

To summarize, there is a local  analytic diffeomorphism   $\Xi:~V\ni (x,z)\mapsto (\tilde x, \tilde z)\in \tilde V$ from an open connected  neighbourhood $V$ of $0$ to another open connected neighbourhood of $0$  such that
\begin{itemize}
  \item for all $(x,z)\in V$, we have $\phi(x,z) >0$ (resp. $<0$, $=0$) when $x>0$ (resp. $<0$, $=0$) .
  \item $\forall (x,z)\in V$, $f(x, z)= -\phi(x,z) -\frac{1}{2} \sum_{k=1}^n (\psi_k(z))^2$.
  \item $\det\left.\begin{pmatrix}
               \dv{\phi}{x} & \dv{\phi}{z} \\
               \dv{\psi}{x} & \dv{\psi}{z} \\
             \end{pmatrix}
             \right|_{(x,z)}
   = \left|\left.\dv{f}{x}\right|_{(0,0)}\right| \sqrt{\left| \det\left(\left.\dvv{f}{z}\right|_{(0,0)}\right) \right| }  ~(1+d(x,z))
   $ where  $d$ is analytic  on $V$ with $d(0,0)=0$.
\end{itemize}

Since $V$ is a neighbourhood of $0$, there exists a $\eta\in(0,1)$ such that  $\mathcal{C}_\eta=(0,\eta)\times (-\eta,\eta)^n \subset V$.  Moreover, there exists $c< f(0,0)$ such that, $\forall (x,z)\in\DD/\mathcal{C}_\eta$, $f(x,z)\leq c$. Since:
\begin{multline*}
 \mathcal{I}_g( \Nm ) = \int_{(x,z)\in\mathcal{C}_\eta}x^m g(x,z) e^{\Nm  f(x,z)} \,  \mathrm dx~\mathrm dz
 +\int_{(x,z)\in\DD/\mathcal{C}_\eta} x^m g(x,z) e^{\Nm  f(x,z)} \,  \mathrm dx~\mathrm dz
 \\
 =  e^{\Nm  f(0,0)}
 \left( \int_{(x,z)\in\mathcal{C}_\eta}x^m g(x,z) e^{\Nm  (f(x,z)-f(0,0))} \,  \mathrm dx~\mathrm dz +
 e^{-\Nm  (f(0,0)-c)} \int_{(x,z)\in\DD/\mathcal{C}_\eta} x^m g(x,z) e^{\Nm  (f(x,z)-c)} \,  \mathrm dx~\mathrm dz
 \right)
 \\
 =
  e^{\Nm  f(0,0)}
 \left( \int_{(x,z)\in\mathcal{C}_\eta}x^m g(x,z) e^{\Nm  (f(x,z)-f(0,0))} \,  \mathrm dx~\mathrm dz +
O\left( e^{-\Nm  (f(0,0)-c)}\right)
 \right)
\end{multline*}
we just have to consider the asymptotic expansion of:
$$
I_\eta(\Nm)= \int_{(x,z)\in\mathcal{C}_\eta}x^m g(x,z) e^{\Nm  (f(x,z)-f(0,0))} \,  \mathrm dx~\mathrm dz
$$
Since $\mathcal{C}_\eta\subset V$, we can consider the change of variable $(\tilde x, \tilde z)=\Xi(x,z)$ that yields:
$$
  I_\eta(\Nm)= \int_{(\tilde x,\tilde z)\in\Xi(\mathcal{C}_\eta)} \tilde x^m \tilde g(\tilde x,\tilde z) e^{ -\Nm \left( \tilde x + \frac{1}{2}\sum_{k=1}^{n} \tilde z_k^2 \right) } \,  \mathrm d\tilde x~\mathrm d\tilde z
$$
where:
$$
\tilde g(\tilde x,\tilde z)=  \frac{ g(\Xi^{-1}(\tilde x,\tilde z))  }{\big(f_1(\Xi^{-1}(\tilde x,\tilde z))\big)^m \left|\left.\dv{f}{x}\right|_{(0,0)}\right| \sqrt{\left| \det\left(\left.\dvv{f}{z}\right|_{(0,0)}\right) \right| } } ~(1+\tilde d(\tilde x,\tilde z)),
$$
and $\tilde d$ is an analytic  function with $\tilde d(0,0)=0$. Since, for all $(\tilde x,\tilde z)\in\Xi(\mathcal{C}_\eta)$ we have $\tilde x\geq 0$, there exists a $\tilde \eta >0$ such that $\widetilde{\mathcal{C}}_{\tilde \eta}=(0,\tilde \eta) \times (-\tilde \eta,\tilde \eta)^n \subset \Xi(\mathcal{C}_\eta)$. Thus, similarly to the passage from
$ \mathcal{I}_g( \Nm ) $ to $I_\eta(\Nm)$, we can just  consider, up to exponentially small terms versus $\Nm$, the asymptotic expansion of:
\begin{equation}\label{eq:Itilde}
\tilde I_{\tilde \eta} = \int_{(\tilde x,\tilde z) \in \widetilde{\mathcal{C}}_{\tilde \eta}}\tilde x^m \tilde g(\tilde x,\tilde z) e^{ -\Nm \left(\tilde x + \frac{1}{2}\sum_{k=1}^{n} \tilde z_k^2 \right) } \,  \mathrm d\tilde x~\mathrm d\tilde z
.
\end{equation}

When $g(0,0)\neq 0$, we have  $\tilde g(0,0) \neq 0$. Set $\tilde g(\tilde x,\tilde z)= \tilde g(0,0) +\tilde x \tilde g_1(\tilde x,\tilde z) +\sum_{k=1}^n \tilde z_k \tilde h_k(\tilde x,\tilde z) $ with $\tilde g_1$ and $\tilde h_k$ bounded analytic  functions  on $\widetilde{\mathcal{C}}_{\tilde \eta}$. We get:
\begin{multline*}
 \tilde I_{\tilde \eta} = \tilde g(0,0) \int_{(\tilde x,\tilde z) \in \widetilde{\mathcal{C}}_{\tilde \eta}}\tilde x^m  e^{ -\Nm \left( \tilde x + \frac{1}{2}\sum_{k=1}^{n} \tilde z_k^2 \right) } \,  \mathrm d\tilde x~\mathrm d\tilde z
 \\
 + \int_{(\tilde x,\tilde z) \in \widetilde{\mathcal{C}}_{\tilde \eta}}\tilde x^{m+1} \tilde g_1(\tilde x,\tilde z)   e^{ -\Nm \left(\tilde x + \frac{1}{2}\sum_{k=1}^{n} \tilde z_k^2 \right) } \,  \mathrm d\tilde x~\mathrm d\tilde z
 \\
 + \int_{(\tilde x,\tilde z) \in \widetilde{\mathcal{C}}_{\tilde \eta}}\tilde x^m \left( \sum_{k=1}^n \tilde z_k \tilde h_k(\tilde x,\tilde z) \right) e^{ -\Nm \left(\tilde x + \frac{1}{2}\sum_{k=1}^{n} \tilde z_k^2 \right) } \,  \mathrm d\tilde x~\mathrm d\tilde z
 .
\end{multline*}
Up to exponentially small terms versus $\Nm$,  the first integral in the right hand-side member can be replaced by:
$$
 \int_{(\tilde x,\tilde z) \in (0,+\infty) \times (-\infty,+\infty)^n} \tilde x^m e^{ -\Nm \left(\tilde x + \frac{1}{2}\sum_{k=1}^{n} \tilde z_k^2 \right) } \,  \mathrm d\tilde x~\mathrm d\tilde z
 = \frac{m!}{\Nm^{m+1}} \left(\frac{2\pi}{\Nm}\right)^{n/2}
 .
 $$
 For the second integral,  $m+1$  integrations by part versus $\tilde x$  are   necessary:
 \begin{multline*}
     \int_{(\tilde x,\tilde z) \in \widetilde{\mathcal{C}}_{\tilde \eta}} \tilde x^{m+1}  \tilde g_1(\tilde x,\tilde z) e^{ -\Nm \left( \tilde x + \frac{1}{2}\sum_{k=1}^{n} \tilde z_k^2 \right) } \,  \mathrm d\tilde x~\mathrm d\tilde z
     \\
     = \int_{\tilde z \in  (-\tilde\eta,\tilde\eta)^n}
     \left(\int_{0}^{\tilde \eta} \tilde x^{m+1} \tilde g_1(\tilde x,\tilde z)
     e^{ -\Nm  \tilde x}\,  \mathrm d\tilde x\right)
     e^{ - \frac{\Nm }{2}\sum_{k=1}^{n} \tilde z_k^2  } \,  \mathrm  d\tilde z,
 \end{multline*}
 where, via  $m+1$  integrations by part,  we get:
 $$
 \int_{0}^{\tilde \eta} \tilde x^{m+1} \tilde g_1(\tilde x,\tilde z)
     e^{ -\Nm  \tilde x}\,  \mathrm d\tilde x = \frac{1}{\Nm^{m+1}} \int_{0}^{\tilde \eta}  \tilde g_{m+2}(\tilde x,\tilde z)
     e^{ -\Nm  \tilde x}\,  \mathrm d\tilde x + O(e^{-\tilde \eta \Nm}/\Nm)
 $$
 with $\tilde g_{m+2} = \frac{\partial ^{m+1}}{\partial \tilde x^{m+1}} \big( \tilde x^{m+1} \tilde g_1(\tilde x,\tilde z)\big)$.
We get:
 \begin{multline*}
     \int_{(\tilde x,\tilde z) \in \widetilde{\mathcal{C}}_{\tilde \eta}} \tilde x^{m+1} \tilde g_1(\tilde x,\tilde z) e^{ -\Nm \left( \tilde x + \frac{1}{2}\sum_{k=1}^{n} \tilde z_k^2 \right) } \,  \mathrm d\tilde x~\mathrm d\tilde z
     \\
     = \frac{1}{\Nm^{m+1}} \int_{(\tilde x,\tilde z) \in \widetilde{\mathcal{C}}_{\tilde \eta}}
      \tilde g_{m+2}(\tilde x,\tilde z)
     e^{ -\Nm \left( \tilde x + \frac{1}{2}\sum_{k=1}^{n} \tilde z_k^2 \right) } \,  \mathrm d\tilde x~\mathrm d\tilde z
     + O(e^{-\tilde \eta \Nm}/\Nm)
     \\
     =  O\left( \frac{1}{\Nm^{m+n/2+2}} \right),
 \end{multline*}
 since $\int_{0}^{\tilde \eta}  \tilde g_{m+2}(\tilde x,\tilde z)
     e^{ -\Nm  \tilde x}\,  \mathrm d\tilde x$ is of order $1/\Nm$.

 Similarly, we get, with $m$ integration by part versus $\tilde x$,
 \begin{multline*}
     \int_{(\tilde x,\tilde z) \in \widetilde{\mathcal{C}}_{\tilde \eta}}\tilde x^{m}  \tilde z_k \tilde h_k(\tilde x,\tilde z) e^{ -\Nm \left(\tilde x + \frac{1}{2}\sum_{l=1}^{n} \tilde z_{l}^2 \right) } \,  \mathrm d\tilde x~\mathrm d\tilde z
     \\
     = \frac{1}{\Nm^{m}} \int_{(\tilde x,\tilde z) \in \widetilde{\mathcal{C}}_{\tilde \eta}}
     \tilde z_k
     \tilde q_{k,m}(\tilde x,\tilde z)
     e^{ -\Nm \left( \tilde x + \frac{1}{2}\sum_{l=1}^{n} \tilde z_{l}^2 \right) } \,  \mathrm d\tilde x~\mathrm d\tilde z
     + O(e^{-\tilde \eta \Nm}/\Nm)
          .
 \end{multline*}
 where $\tilde q_{k,m}(\tilde x,\tilde z)=  \frac{\partial ^{m}}{\partial \tilde x^{m}} \big( \tilde x^{m} \tilde h_k(\tilde x,\tilde z)  \big)$.
  A single integration by part versus $\tilde z_k$ yields:
 \begin{multline*}
     \int_{(\tilde x,\tilde z) \in \widetilde{\mathcal{C}}_{\tilde \eta}}  \tilde z_k \tilde q_{k,m}(\tilde x,\tilde z) e^{ -\Nm \left(\tilde x + \frac{1}{2}\sum_{l=1}^{n} \tilde z_{l}^2 \right) } \,  \mathrm d\tilde x~\mathrm d\tilde z
     \\
     = \frac{1}{\Nm} \int_{(\tilde x,\tilde z) \in \widetilde{\mathcal{C}}_{\tilde \eta}}
      \dv{\tilde q_{k,m}}{\tilde z_k}(\tilde x,\tilde z)
     e^{ -\Nm \left(\tilde x  + \frac{1}{2}\sum_{l=1}^{n} \tilde z_{l}^2 \right) } \,  \mathrm d\tilde x~\mathrm d\tilde z
     + O(e^{-\tilde \eta^2 \Nm/2}/\Nm)
     .
 \end{multline*}
     This implies that:
\begin{multline*}
     \int_{(\tilde x,\tilde z) \in \widetilde{\mathcal{C}}_{\tilde \eta}}\tilde x^{m}  \tilde z_k \tilde h_k(\tilde x,\tilde z) e^{ -\Nm \left(\tilde x + \frac{1}{2}\sum_{kl=1}^{n} \tilde z_{l}^2 \right) } \,  \mathrm d\tilde x~\mathrm d\tilde z
     \\
     = \frac{1}{\Nm^{m+1}} \int_{(\tilde x,\tilde z) \in \widetilde{\mathcal{C}}_{\tilde \eta}}
      \dv{\tilde r_{k,m}}{\tilde z_k}(\tilde x,\tilde z)
     e^{ -\Nm \left( \tilde x + \frac{1}{2}\sum_{l=1}^{n} \tilde z_{l}^2 \right) } \,  \mathrm d\tilde x~\mathrm d\tilde z
    +  O(e^{-\tilde \eta^2 \Nm/2}/\Nm)
    \\
        =  O\left( \frac{1}{\Nm^{m+n/2+2}} \right)
          .
 \end{multline*}
 Thus, we get~\eqref{case1:thm:boundary}, thanks to
 $\tilde I_{\tilde \eta} =   \frac{\tilde g(0,0) m!}{\Nm^{m+1}} \left(\frac{2\pi}{\Nm}\right)^{n/2} \left( 1 + O(1/\Nm)\right)$ and $\tilde g(0,0)=   \frac{ g(0,0)}{ \left( \left|\left.\dv{f}{x}\right|_{(0,0)}\right|\right)^{m+1}\sqrt{\left| \det\left(\left.\dvv{f}{z}\right|_{(0,0)}\right) \right| } } $

 Assume now that  $g(0,0)=0$, $\left.\dv{g}{x}\right|_{(0,0)}=0$ and $\left.\dv{g}{z}\right|_{(0,0)}=0$. Consider then the function  $\tilde g$ in ~\eqref{eq:Itilde}. We have  $\tilde g(0,0)=0$, $\left.\dv{\tilde g}{\tilde x}\right|_{(0,0)}=0$ and $\left.\dv{\tilde g}{\tilde z}\right|_{(0,0)}=0$. Moreover, denoting:
 $$
 \lambda_0=\sqrt{\left| \det\left(\left.\dvv{f}{z}\right|_{(0,0)}\right) \right|~ } \left( -\left.\dv{f}{x}\right|_{(0,0)} \right)^{m+1},
 $$
  we have:
 $$\lambda_0 \tilde g(\phi(x,z),\psi(z)) = g(x,z) ~(1+ e(x,z)),$$
  with $e$ an analytic   function with  $e(0,0)=0$. Similarly to~\eqref{eq:tracegInt}, we get:
 \begin{equation}\label{eq:traceg}
 \tr{\left.\dvv{\tilde g}{\tilde z}\right|_{(0,0)}} = \frac{\tr{-\left.\dvv{g}{z}\right|_{(0,0)} \left(\left.\dvv{f}{z}\right|_{(0,0)}\right)^{-1} }}{\sqrt{\left| \det\left(\left.\dvv{f}{z}\right|_{(0,0)}\right) \right|~ } \left( -\left.\dv{f}{x}\right|_{(0,0)} \right)^{m+1}}.
\end{equation}
Since $\tilde g$ and its first partial derivatives versus $\tilde x$ and $\tilde z_k$ vanish, we have:
$$
\tilde g(\tilde x,\tilde z)=
 \tilde x^2  \tilde a (\tilde x,\tilde z)
+ \sum_{k,k'=1}^n \tilde z_k \tilde z_{k'} \tilde b_{k,k'} (\tilde x,\tilde z)
+ \sum_{k=1}^n \tilde x \tilde z_{k} \tilde c_{k} (\tilde x,\tilde z)
$$
where the function $\tilde a$, $\tilde b_{k,k'}$ and $\tilde c_{k}$ are analytic . To evaluate the integral in~\eqref{eq:Itilde}, we have to consider  the dominant terms of three kinds of integrals:
\begin{align*}
  A= \int_{(\tilde x,\tilde z) \in \widetilde{\mathcal{C}}_{\tilde \eta}} \tilde x^{m+2} \tilde a (\tilde x,\tilde z)  e^{ -\Nm \left( \tilde x + \frac{1}{2}\sum_{l=1}^{n} \tilde z_l^2 \right) } \,  \mathrm d\tilde x~\mathrm d\tilde z,
  \\
    B_{k,k'}= \int_{(\tilde x,\tilde z) \in \widetilde{\mathcal{C}}_{\tilde \eta}} \tilde x^m  \tilde z_k \tilde z_{k'} \tilde b_{k,k'} (\tilde x,\tilde z)  e^{ -\Nm \left( \tilde x + \frac{1}{2}\sum_{l=1}^{n} \tilde z_l^2 \right) } \,  \mathrm d\tilde x~\mathrm d\tilde z,
  \\
  C_{k}= \int_{(\tilde x,\tilde z) \in \widetilde{\mathcal{C}}_{\tilde \eta}} \tilde x^{m+1}  \tilde z_{k} \tilde c_{k} (\tilde x,\tilde z)  e^{ -\Nm \left( \tilde x + \frac{1}{2}\sum_{l=1}^{n} \tilde z_l^2 \right) } \,  \mathrm d\tilde x~\mathrm d\tilde z
  .
\end{align*}
As done previously, $m+2$ integrations by part on $\tilde x$   yield  $A= O\left(\Nm^{-m-n/2-3}\right)$. As done previously, $m+1$ integrations by part versus $\tilde x$ and a single integration by part  versus $\tilde z_k $ provide
 $C_{k}= O\left(\Nm^{-m-n/2-3}\right)$. For $k\neq k'$, $m$ integrations by part versus $\tilde x$, one integration by part versus $\tilde z_k$ followed by another one versus $\tilde z_{k'}$, yield  similarly   to  $B_{k,k'}= O\left(\Nm^{-m-n/2-3}\right)$. For $k=k'$, we start with $m$ integrations by part versus $\tilde x$:

  \begin{multline*}
B_{k,k}= \int_{(\tilde x,\tilde z) \in \widetilde{\mathcal{C}}_{\tilde \eta}} \tilde x^m \tilde z_k^2 \tilde b_{k,k} (\tilde x,\tilde z)  e^{ -\Nm \left( \tilde x + \frac{1}{2}\sum_{l=1}^{n} \tilde z_l^2 \right) } \,  \mathrm d\tilde x~\mathrm d\tilde z
\\
      = \frac{1}{\Nm^{m}} \int_{(\tilde x,\tilde z) \in \widetilde{\mathcal{C}}_{\tilde \eta}}
     \tilde z_k^2 \tilde q_{k,m}(\tilde x,\tilde z)
     e^{ -\Nm \left( \tilde x + \frac{1}{2}\sum_{k=1}^{n} \tilde z_k^2 \right) } \,  \mathrm d\tilde x~\mathrm d\tilde z
     + O(e^{-\tilde \eta \Nm}/\Nm)
          .
 \end{multline*}
 where $\tilde q_{k,m}(\tilde x,\tilde z)=  \frac{\partial ^{m}}{\partial \tilde x^{m}} \big( \tilde x^{m}  \tilde b_{k,k} (\tilde x,\tilde z)  \big)$.
 We can notice that  $\tilde q_{k,m}(0)= m! \tilde b_{k,k}(0)$. A single integration by part versus $\tilde z_k$ yields to:
 \begin{multline*}
 \int_{(\tilde x,\tilde z) \in \widetilde{\mathcal{C}}_{\tilde \eta}}  \tilde z_k^2 \tilde q_{k,m} (\tilde x,\tilde z)  e^{ -\Nm \left( \tilde x + \frac{1}{2}\sum_{l=1}^{n} \tilde z_l^2 \right) } \,  \mathrm d\tilde x~\mathrm d\tilde z
\\
= \frac{1}{\Nm} \int_{(\tilde x,\tilde z) \in \widetilde{\mathcal{C}}_{\tilde \eta}}\left( \tilde q_{k,m} (\tilde x,\tilde z)+ \tilde z_k \dv{\tilde q_{k,m}}{\tilde z_k} (\tilde x,\tilde z)\right)  e^{ -\Nm \left(\tilde x + \frac{1}{2}\sum_{l=1}^{n} \tilde z_l^2 \right) } \,  \mathrm d\tilde x~\mathrm d\tilde z + O(e^{-\Nm\tilde \eta^2/2})
\\
=
 \tilde q_{k,m}(0)\frac{1}{\Nm^{2}} \left(\frac{2\pi}{\Nm}\right)^{n/2}
+ O\left(\Nm^{-n/2-3}\right)
.
 \end{multline*}
With $\tilde q_{k,m}(0)= m! \tilde b_{k,k}(0)$, the sum $ A + \sum _k C_k + \sum_{k,k'} B_{k,k'}$  corresponding  the integral in~\eqref{eq:Itilde} reads:
$$
\tilde I_{\tilde \eta}(\Nm) = \frac{\sum_{k=1}^{n} m! \tilde b_{k,k}(0)}{\Nm^{m+2}} \left(\frac{2\pi}{\Nm}\right)^{n/2}
+ O\left(\Nm^{-m-n/2-3}\right)
.
$$
Since up to exponentially small terms,  $ \tilde I_{\tilde \eta}$ and $e^{-\Nm f(0)} \mathcal{I}_g( \Nm ) $ coincide, we  get~\eqref{case2:thm:boundary} using~\eqref{eq:traceg},  since
$
\sum_{k=1}^{n} \tilde b_{k,k}(0)=\tfrac{1}{2} \tr{\left.\dvv{\tilde g}{\tilde z}\right|_{0}}
$.  
\end{proof}

The asymptotic expansions of theorems~\ref{thm:interior} and~\ref{sec:boundary} yield  directly the following approximations of the  Bayesian mean  and variance.

\begin{cor} \label{cor:BayesianMeanVariance}
Consider the analytic  function $f(z)$ of theorem~\ref{thm:interior}.    Then we have the following asymptotic for any analytic  function $g(z)$:
\begin{equation}
\mathcal{M}_g( \Nm )\triangleq \frac{ \int_{z\in(-1,1)^{n}} g(z) \exp\left( \Nm  f(z)\right) \mathrm dz}{ \int_{z\in(-1,1)^{n}}  \exp\left( \Nm  f(z)\right) \mathrm dz} = g(0) +   O(\Nm^{-1}) \label{eq:AsymptoticMeanInterior}
\end{equation}
We have also:
\begin{multline} \label{eq:AsymptoticVarianceInterior}
  \mathcal{V}_g( \Nm )\triangleq \frac{ \int_{z\in(-1,1)^{n}}  \Big(g(z)-\mathcal{M}_g( \Nm )\Big)^2 \exp\left( \Nm  f(z)\right) ~\mathrm dz}{ \int_{z\in(-1,1)^{n}}   \exp\left( \Nm  f(z)\right) \, \mathrm  dz}
  \\ = \frac{\tr{-\left.\dvv{g}{z}\right|_{0} \left(\left.\dvv{f}{z}\right|_{0}\right)^{-1} }}{2 \Nm} + O\big(\Nm^{-2}\big)
  .
  \end{multline}

Consider the analytic function $f(x,z)$ of theorem~\ref{thm:boundary}.    Then, we have the following asymptotic for any analytic  function $g(x,z)$:
\begin{equation}
\mathcal{M}_g( \Nm )\triangleq \frac{ \int_{x\in (0,1)} \int_{z\in(-1,1)^{n}} x^m g(x,z) \exp\left( \Nm  f(x,z)\right) \,  \mathrm dx~\mathrm dz}{\int_{x\in (0,1)} \int_{z\in(-1,1)^{n}} x^m  \exp\left( \Nm  f(x,z)\right) \,  \mathrm dx~\mathrm dz} = g(0,0) +   O(\Nm^{-1}) \label{eq:AsymptoticMean}
\end{equation}
We have also:
\begin{multline} \label{eq:AsymptoticVariance}
  \mathcal{V}_g( \Nm )\triangleq \frac{ \int_{x\in (0,1)} \int_{z\in(-1,1)^{n}}x^m  \Big(g(x,z)-\mathcal{M}_g( \Nm )\Big)^2 \exp\left( \Nm  f(x,z)\right) \,  \mathrm dx~\mathrm dz}{\int_{x\in (0,1) } \int_{z\in(-1,1)^{n}} x^m  \exp\left( \Nm  f(x,z)\right) \,  \mathrm dx~\mathrm dz}
  \\ = \frac{\tr{-\left.\dvv{g}{z}\right|_{(0,0)} \left(\left.\dvv{f}{z}\right|_{(0,0)}\right)^{-1} }}{2 \Nm} + O\big(\Nm^{-2}\big)
  .
  \end{multline}
\end{cor}

During the proof of theorem~\ref{thm:interior}, we have proved during the passage from $z$ yo $\tilde z$ coordinates the following lemma.
\begin{lem} \label{lem:trace}
Take two $C^2$ real-value  functions $f$ and $g$ of  $z\in\mathbb{R}^n$. Assume that $0$ is a  regular critical point of $f$  and just a critical  point of $g$.
Take any $C^2$  diffeomorphism $\phi$ defined locally around $0$:   $\tilde z = \phi(z)$. Then:
$$
\tr{-\left.\dvv{g}{z}\right|_{0} \left(\left.\dvv{f}{z}\right|_{0}\right)^{-1} }
=  \tr{-\left.\dvv{\tilde g}{\tilde z}\right|_{\phi(0)} \left(\left.\dvv{\tilde f}{\tilde z}\right|_{\phi(0)}\right)^{-1} }
$$
where $\tilde f(\phi(z))= f(z)$ and $\tilde g(\phi(z))=g(z)$.
\end{lem}
This lemma just says that the above trace formula  is coordinate-free, i.e.,  independent of the local coordinates chosen to compute the Hessian of $f$ and $g$ at their common critical point.

\section{Application to quantum state tomography} \label{sec:Qtomo}
As explained in~\cite{SixPRA2016}, the parameter $p$ to estimate  corresponds to a density operator $\rho$ (quantum state),  a square matrix with complex entries and  belonging to the convex compact set $\DD$ formed by Hermitian $d\times d$ non-negative matrices of trace one.  Then, the  log-likelihood function  admits  the following structure:
\begin{equation}\label{eq:f}
  f(\rho) =  \sum_{\mu \in \mathcal{M}} \log\left( \tr{\rho Y_\mu} \right)
\end{equation}
where the set $\mathcal{M}$ is finite and  each measurement data  $Y_\mu$ belongs also to $\DD$.  For any Hermitian $d\times d$ matrix $A$, (a quantum observable) we are interested to provide an approximation  of Bayesian estimate of  $\tr{\rho A}$,
\begin{equation}\label{eq:Amean}
I_A(\Nm)= \frac{\int_{\DD} \tr{\rho A} e^{\Nm f(\rho)}~ \mathbb{P}_0(\rho) ~\mathrm d\rho}{\int_{\DD} e^{\Nm f(\rho)} ~\mathbb{P}_0(\rho)~\mathrm d\rho}
,
\end{equation}
and of  the  Bayesian variance:
\begin{equation}\label{eq:Avar}
V_A(\Nm)= \frac{\int_{\DD} \Big(\tr{\rho A}-I_A(\Nm)\Big)^2  e^{\Nm f(\rho)}~ \mathbb{P}_0(\rho) ~\mathrm d\rho}{\int_{\DD} e^{\Nm f(\rho)} ~\mathbb{P}_0(\rho)~\mathrm d\rho}
.
\end{equation}
Here above $\mathrm d\rho$ stands for the standard Euclidian volume element on $\DD$,  derived from the Frobenius product between $n\times n$ Hermitian matrices, and $\mathbb{P}_0 >0$ is a probability density  on $\rho$ prior to the measurement data  $(Y_\mu)$.  Since the  number of real parameters to describe  $\rho$ is large in general, it is difficult to compute these integrals even numerically via Monte-Carlo method.

The following lemma  provides  a unitary invariance  characterization  of  any $\bar\rho$ argument of the maximum of $f$ on $\DD$.
   \begin{lem}\label{lem:Qtomo1}
Assume that the  $d\times d$ Hermitian matrix  $\brho$ is an argument of the maximum of $f: \DD\ni\rho\mapsto f(\rho) \in [-\infty,0]$ defined in~\eqref{eq:f} over $\DD$ (the set of density operators). Then necessarily, $\brho$ satisfies the following condition:
\begin{itemize}
  \item $\tr{\brho Y_\mu} >0$ for each $\mu\in\mathcal{M}$;
  \item $\left[ \brho ~, ~\left.\nabla f\right|_{\brho}\right] = \brho \cdot  \left.\nabla f\right|_{\brho} -\left.\nabla f\right|_{\brho} \cdot \brho= 0$, where $\left.\nabla f\right|_{\brho}=  \sum_{\mu \in \mathcal{M}} \frac{Y_\mu}{\tr{\brho Y_\mu}}$ is the gradient of $f$ at $\brho$ for the Frobenius scalar product;

  \item there exists $\bar\lambda >  0$ such that
$\blambda \bP = \bP  \left.\nabla f\right|_{\brho}$ and $ \left.\nabla f\right|_{\brho} \leq \blambda I $,
where $\bP$ is the orthogonal projector on the range of $\brho$ and $I$ is the identity operator.
\end{itemize}
These conditions are also sufficient and characterize the unique maximum  when, additionally,  the vector space spanned by the $Y_\mu$'s coincides with the set of Hermitian matrices.
\end{lem}
\begin{proof} Since $f$ is a concave function of $\rho$, we can use the  standard optimality criterion for  a convex  optimization problem (see, e.g., \cite[section 4.2.3]{BoydBook2009}):   $\brho$   maximizes $f$  over the convex compact  set $\DD$, if and only if,  $\rho\in\DD$, $\tr{(\rho-\brho) \left.\nabla f\right|_{\brho}} \leq 0$.

   Assume that $f(\brho)$ is maximum. Since $f(I) > -\infty$, for each $\mu$ we have  $\tr{\brho Y_\mu} >0$.
   Take $\rho=e^{-i H} \brho e^ {iH}$, where $H$ is an arbitrary Hermitian operator. We have:
   $$
   \tr{e^{-i H} \brho e^ {iH} \left.\nabla f\right|_{\brho}} \leq \tr{\brho  \left.\nabla f\right|_{\brho}} .
   $$
   For $H$ close to zero, we have via the Baker-Campbell-Hausdorff formula, $e^{-i H} \brho e^ {iH} = \brho - i[H,\brho] + O(\tr{H^2})$.
   The above inequality implies that for all $H$ small enough, $\tr{[H,\brho] \left.\nabla f\right|_{\brho}} =\tr{H \left[\brho,\left.\nabla f\right|_{\brho}\right]}=0$ and thus
   $\brho$ and $\left.\nabla f\right|_{\brho}$ commute.

   Consider the spectral decomposition $\brho= U \overline{\Delta} U^\dag$ where $U$ is unitary and $\overline{\Delta}$ diagonal with entries $0\leq \overline{\Delta}_1\leq \overline{\Delta}_2 \leq \ldots\leq \overline{\Delta}_d\leq 1$. Since $\brho$ and $\left.\nabla f\right|_{\brho}$ commute, we have also $\left.\nabla f\right|_{\brho}= U \overline{\Lambda} U^\dag$ with $\overline{\Lambda}$ diagonal  with entries $(\overline{\Lambda}_k)$ Since $\nabla f$ is non negative, these entries are non-negative too. Take  $\rho= U  \Delta U^\dag$ where $\Delta$ is any diagonal matrix with  non negative entries and of trace one. We have:
   $$
   \tr{(\rho-\brho) \left.\nabla f\right|_{\brho}} = \tr{(\Delta-\overline{\Delta})\overline{\Lambda}} \leq 0.
   $$
   This means that, for any $(\Delta_1,\ldots, \Delta_d)\in[0,1]^d$ such that $\sum_{k=1}^{d} \Delta_k=1$ we have:
   $$
   \sum_{k=1}^{d} (\Delta_k -\overline{\Delta}_k) \overline{\Lambda}_k  \leq 0
   .
   $$
   Take $\epsilon >0$, $(k_1,k_2)\in\{1,\ldots,d\}^2$ such that $\overline{\Delta}_{k_1} >0$ and  $k_2\neq k_1$. For $k\in\{1,\ldots,d-1\}/\{k_1,k_2\}$ set   $\Delta_k= \overline{\Delta}_k$ and take
   $\Delta_{k_1}= \overline{\Delta}_{k_1}- \epsilon$ with  $\Delta_{k_2}= \overline{\Delta}_{k_2}+\epsilon$. By construction $\tr{\Delta}=1$ and,  for $\epsilon>0 $ small enough,  $\Delta_k\geq 0$ for all $k\in\{1,\ldots,d\}$. The previous  inequality implies that:
   $$
 \forall (k_1,k_2)\in\{1,\ldots,d\}^2 \text{ such that } \overline{\Delta}_{k_1} >0 \text{ and }k_1\neq k_2, \quad
 \overline{\Lambda}_{k_2}\leq \overline{\Lambda}_{k_1}
 .
   $$
 Thus for all $k_1,k_2$ such that $\overline{\Delta}_{k_1}>0$ and $\overline{\Delta}_{k_2}>0$,   $ \overline{\Lambda}_{k_1}=\overline{\Lambda}_{k_2}=\overline{\lambda} \geq 0$.  For  $k_1,k_2$ such that $\overline{\Delta}_{k_1}>0$ and $\overline{\Delta}_{k_2}=0$,   we have also $ \overline{\Lambda}_{k_2} \leq  \overline{\Lambda}_{k_1} =\overline{\lambda}$.  Thus we get
 $ \overline{\Lambda} \leq \overline{\lambda} I $. With $\overline{\Theta}$ the diagonal matrix of entries $\overline{\Theta}_k=0$ (resp. $=1$)  when $\overline{\Delta}_k=0$ (resp. $>0$), we have $\bP= U \overline{\Theta} U^\dag$ we get
 $ \overline{\lambda} \bP = \bP \left.\nabla f\right|_{\brho} $. Since $\left.\nabla f\right|_{\brho}$ is non negative and cannot be zero,  we have $\overline{\lambda} >0$.

Take   $\brho$  satisfying the conditions of lemma~\ref{lem:Qtomo1}.  Since they are  unitary invariant, we can assume that $\brho$ and $\left.\nabla f\right|_{\brho}$ are diagonal operators $\overline{\Delta}$ and  $\overline{\Lambda}$. Since we are in the convex situation, it is enough to prove that $\brho$ is a local maximum. Any local variation of $\rho$ around $\brho$ and remaining inside $\DD$ is  parameterized   via the following mapping:
 $$
 (H,D) \mapsto e^{-i H}( \overline{\Delta} + D) e^{i H}=\rho_{H,D}
 $$
 where $H$ is any  Hermitian matrix and $D$ is any  diagonal matrix of zero trace such that $\overline{\Delta} + D\geq 0$.
 We have the following expansion for $H$ and $D$ around zero:
 $$
 \rho_{H,D}=\overline{\Delta} + D -i[H,\overline{\Delta}] - i[H,D] -\tfrac{1}{2} [H,[H,\overline{\Delta}]] + O(\tr{H^3+D^3})
 .
 $$
 This yields to the following second order expansion of $(H,D)\mapsto f(\rho_{H,D})$ around zero:
 \begin{multline*}
   f(\rho_{H,D}) = f(\brho) + \tr{\overline{\Lambda} \left(D -i[H,\overline{\Delta}] - i[H,D] -\tfrac{1}{2} [H,[H,\overline{\Delta}]]\right)}
   \\-  \sum_{\mu\in\mathcal{M}}
  \frac{\trr{( \rho_{H,D}-\brho)Y_\mu}}{2 \trr{\brho Y_\mu}} + O(\| \rho_{H,D}-\brho\|^3)
   .
 \end{multline*}
  By assumptions, $\overline{\Lambda}$,  $\overline{\Delta} $ and $D$ are diagonal. Thus  $\tr{\overline{\Lambda} \left(-i[H,\overline{\Delta}] - i[H,D] \right)}=0$.
  Some elementary arguments exploiting $\blambda \overline{\Theta} \leq  \overline{\Delta} \leq \blambda I$,  show that $ \tr{\overline{\Lambda}D}\leq 0$ since $D$ is such that $\overline{\Delta}+D$ is nonnegative and of  trace one.
  We also have:
  $$
 - \tr{\overline{\Lambda} \left( [H,[H,\overline{\Delta}]]\right)}= \tr{[H,\overline{\Lambda}] ~[H,\overline{\Delta}]}
 =- 2 \sum_{k_1\in P ,k_2\in Q } \overline{\Delta}_{k_1}  \big(\overline{\lambda}-\overline{\Lambda}_{k_2}\big) |H_{k_1k_2}|^2
 \leq 0
 $$
 where $P=\{k~|~\overline{\Delta}_k >0\}$ and $Q=\{k~|~\overline{\Delta}_k =0\}$.

 Consequently:
 $$
 f(\rho_{H,D}) \leq f(\brho)  -  \sum_{\mu\in\mathcal{M}} \frac{\trr{( \rho_{H,D}-\brho) Y_\mu}}{2 \trr{\brho Y_\mu}} + O(\| \rho_{H,D}-\brho\|^3)
 .
 $$
 Since  the vector space spanned by the $Y_\mu$  coincide with the set of Hermitian matrices,  the quadratic form $X \mapsto \sum_{\mu\in\mathcal{M}} \frac{\trr{X Y_\mu}}{2 \trr{\brho Y_\mu}}$ is non-degenerate ($X$ is any Hermitian matrix) and  $f$ is strongly concave. Thus  we have  $f(\rho) <  f(\brho)$  for $\rho \neq \brho$ close to $\brho$. Consequently, $\brho$ is a strict local maximum and  this  maximum  is unique and global  since $f$ is concave.

\end{proof}

   \begin{thm}\label{thm:Qtomo2}
Consider the  log-likelihood function $f$ defined in~\eqref{eq:f}. Assume that  the $Y_\mu$'s  span  the set of Hermitian matrices. Denote by $\brho$ the  unique maximum of $f$ on $\DD$ and define a projector $\overline{P}$ such that, in addition to the necessary and sufficient conditions of  lemma~\ref{lem:Qtomo1}, we have $\ker\left(\overline{\lambda} I - \left.\nabla f\right|_{\brho}\right)= \ker (I-\overline{P})$. Then, for any Hermitian operator $A$, its Bayesian mean defined in~\eqref{eq:Amean} admits the following asymptotic expansion
$$
I_A(\Nm)= \tr{A \brho} + O(1/\Nm)
$$
and its Bayesian variance defined in~\eqref{eq:Avar} satisfies
$$
V_A(\Nm)= \tr{A_{\parallel} ~\left(\bF\right)^{-1}\!\!(A_\parallel)}/\Nm + O(1/\Nm^2)
$$
where
\begin{itemize}
  \item for any Hermitian operator $B$, $B_{\parallel}$ stands for is orthogonal projection on the tangent space  at $\brho$ to the submanifold of Hermitian matrices with a rank equal to the rank of  $\brho$ and of unit trace. It  reads
   \begin{equation} \label{eq:ProjOrtho}
     B_{\parallel}= B - \frac{\tr{B \overline{P}}}{\tr{\overline{P}}} \overline{P} - (I-\overline{P}) B (I-\overline{P});
   \end{equation}
   when $\brho$ is full rank, $ B_{\parallel}= B - \tr{B}I/d$ since $\overline{P}=I$;

  \item the linear super-operator $\bF$  corresponds to the Hessian at $\brho$  of the restriction of $f$ to the manifold of Hermitian matrices of rank equal to the rank of $\brho$ and with  trace one. Its  reads for any Hermitian operator $X$,
\begin{equation}\label{eq:FisherInfo}
  \bF(X)= \sum_{\mu} \frac{ \tr{ X Y_{\mu\parallel} }}{\trr{ \brho Y_\mu}} Y_{\mu\parallel}
+  \left(\overline{\lambda}I-\left.\nabla f\right|_{\brho}\right) X \brho^{+} +   \brho^{+} X \left(\overline{\lambda}I-\left.\nabla f\right|_{\brho}\right)
\end{equation}
with   $\brho^{+}$ the Moore-Penrose pseudo-inverse of $\brho$; the  restriction of $X \mapsto \tr{X \bF(X)}$  to the tangent space at $\brho$  is positive definite; thus the restriction of $\bF$ to this tangent space is invertible and can be seen as the analogue of the Fisher information; its inverse at $A_\parallel$  is denoted here above  by $\left(\bF\right)^{-1}\!\!(A_\parallel)$.
\end{itemize}
   \end{thm}
   \begin{proof}
The Hessian of $f$ at  $\rho\in\DD$ where $f(\rho)> -\infty$  reads:
$$
\left.\nabla^2 f\right|_{\rho} (X,Z)= - \sum_{\mu} \frac{\tr{X Y_\mu} \tr{Z Y_\mu}}{\trr{\rho Y_\mu}}
$$
where $X$ and $Z$ are any Hermitian matrices.  Since it is positive definite, $f$ is strongly concave. Consequently  the argument of maximum of $f$ on $\DD$ is unique, denoted $\brho$ and satisfies the condition of lemma~\ref{lem:Qtomo1}.
Take a small neighbourhood $\VV$  of $\brho$ in $\DD$. Then there exists a $\epsilon >0$  such that, for $\rho\in\DD/\VV$, $f(\rho) \leq f(\brho) - \epsilon$.
 To investigate $\int_{\VV} e^{\Nm (f(\rho)-f(\brho))}~ \mathbb{P}_0(\rho) ~\mathrm d\rho$, we consider the following local coordinates based on the spectral decomposition of $\brho= U \overline{\Delta} U^\dag$ with $U$ unitary and $\overline{\Delta}$ diagonal with  entries $0 = \overline{\delta}_1\leq \overline{\delta}_2 \leq \ldots\leq \ldots, {\overline{\delta}}_d \leq 1$ with $\sum_{k=1}^d \overline{\delta}_k=1$.
Denote by $r$ the rank of $\brho$ and assume that $r < d$ (the case $r=d$ is much simpler, relies on theorem~\ref{thm:interior} and is left to the reader). We have $\overline{\delta}_k=0$ for $k$ between $1$ and $d-r$ and $\overline{\delta}_k>0$ for $k$ between $d-r+1$ and $d$.
Since the volume element $\mathrm d\rho $  used in~\eqref{eq:Amean} and~\eqref{eq:Avar} is unitary invariant  \cite[page 42]{MehtaBook2004}, we can assume without lost of generality  that $\brho$ is diagonal (change $\mathbb{P}_0(\bullet)$ to $\mathbb{P}_0(U \bullet  U^\dag)$ and replace each $Y_\mu$ by $U^\dag Y_\mu U$  in the definition~\eqref{eq:f} of $f$). Consider the following map
$$
(\xi,\zeta,\omega) \mapsto \Upsilon=
\exp \left(    \begin{bmatrix} 0 & \omega\\  -\omega^\dag & 0\end{bmatrix} \right)
     \begin{bmatrix}\qquad  \xi\qquad  &  0 \\ 0 & \overline{\Delta}_r + \zeta - \frac{\tr{\xi}}{r} I_r \end{bmatrix}
     \exp   \left( \begin{bmatrix} 0 & -\omega\\  \omega^\dag  &0\end{bmatrix} \right)
$$
where $\xi$ is a $(d-r)\times (d-r)$ Hermitian matrix , $\omega$ is $(d-r)\times r$ matrix with complex entries, $\zeta$ is a $r\times r$ Hermitian matrix of trace $0$, $I_r$ is the identity matrix of size $r$ and $\overline{\Delta}= \begin{bmatrix} 0 &  0 \\ 0 & \overline{\Delta}_r  \end{bmatrix}$. This map is a local diffeomorphism  from a neighbourhood of $(0,0,0)$ to a neighbourhood of $\brho$ in the set of Hermitian matrices of trace one since its tangent map at zero, given by:
\begin{equation} \label{eq:tangentmap}
  (\delta \xi,\delta\zeta, \delta\omega) \mapsto  \begin{bmatrix}\qquad  \delta \xi\qquad  &  \delta\omega~ \overline{\Delta}_r  \\ \overline{\Delta}_r~ \delta\omega^\dag &
\delta\zeta - \frac{\tr{\delta\xi}}{r} I_r\end{bmatrix}= \delta\rho
\end{equation}
is bijective (local inversion theorem).   Thus, we have:
\begin{multline*}
  \int_{\VV} e^{\Nm (f(\rho)-f(\brho))}~ \mathbb{P}_0(\rho) ~\mathrm d\rho
  \\ = \int_{ \Upsilon^{-1}(\VV)} e^{\Nm (f(\xi,\zeta,\omega) - f(0,0,0))} \mathbb{P}_0(\xi,\zeta,\omega) J(\xi,\zeta,\omega) \mathrm  d\xi ~\mathrm d\zeta~\mathrm d\omega
\end{multline*}
where $f(\xi,\zeta,\omega)$ and $\mathbb{P}_0(\xi,\zeta,\omega)$ stand for $f(\Upsilon(\xi,\zeta,\omega))$ and $\mathbb{P}_0(\Upsilon(\xi,\zeta,\omega))$ and where
$J(\xi,\zeta,\omega)$ is the Jacobian of this change of coordinates.

Since the  constraint $\Upsilon(\xi,\zeta,\omega) \geq 0$ reads $\xi \geq 0$, we consider another change of variables to parameterize $\xi \geq 0$ around $0$:  $\Xi:~(x,\sigma,\zeta,\omega) \mapsto  (x \sigma=\xi, \zeta,\omega)  $, where $ x\geq 0$ and $\sigma$ is a $(d-r)\times (d-r)$  density matrix.
Then:
\begin{multline*}
  \int_{\VV} e^{\Nm (f(\rho)-f(\brho))}~ \mathbb{P}_0(\rho) ~\mathrm d\rho
  \\ = \int_{ \Xi^{-1}\big(\Upsilon^{-1}(\VV)\big)} e^{\Nm (f(x\sigma,\zeta,\omega) - f(0,0,0))} \mathbb{P}_0(x\sigma,\zeta,\omega) J(x\sigma,\zeta,\omega)   x^m ~\mathrm dx  ~\mathrm d\sigma ~\mathrm d\zeta~\mathrm d\omega
\end{multline*}
with $m=(d-r+1)(d-r-1)$.  This change of variables is singular, since for $x=0$ it is not invertible. Nevertheless,  the set of coordinates verifying $x=0$ is of zero measure,  and then this has no impact on the integral.
Take $\eta > 0$ small enough  and  adjust the  neighbourhood $\VV$ of $\brho$ such that
$
\Xi^{-1}\big(\Upsilon^{-1}(\VV)\big)$ coincides with the set where $x\in(0,\eta)$, $\sigma\in\DD_{d-r}$ and  all the  real  and imaginary parts of  $\zeta$ and $\omega$ entries belong to $(-\eta,\eta)$.  Following the notations of theorem~\ref{thm:interior}, set $z=(\zeta,\omega)$.
We have $z \in(-\eta,\eta)^n$ with $n=2r(d-r) + (r+1)(r-1)$ and:
\begin{multline*}
  \int_{\VV} e^{\Nm (f(\rho)-f(\brho))}~ \mathbb{P}_0(\rho) ~\mathrm d\rho
\\ = \int_{\sigma\in\DD_{d-r}} \left(\int_{(x,z)\in (0,\eta)\times (-\eta,\eta)^n} e^{\Nm f(x\sigma, z)} x^m J(x\sigma,z) \mathbb{P}_0(x\sigma,z)~\mathrm dx~\mathrm dz \right)  \mathrm d\sigma
.
\end{multline*}
For each $\sigma\in\DD_{d-r}$, let us    use~\eqref{case1:thm:boundary}, with $J(x\sigma,z) \mathbb{P}_0(x\sigma,z)$ standing for $g(x,z)$.
We have $g(0,0)= J(0,0) \mathbb{P}_0(0,0) >0$.
By construction, we have:
$$
 f(x\sigma, z) =  x f_1(x,\sigma,z) + f(0,z)
$$
where $f_1(x,\sigma,z)$ is analytic versus $(x,z)$ and $f_1(0,\sigma,0)=  \left( \tr{\Lambda_{d-r}\sigma} -\overline{\lambda}\right)  $. This is
based on~\eqref{eq:tangentmap} and  on the diagonal structure  $\left.\nabla f\right|_{\brho}=\begin{bmatrix} \Lambda_{d-r} &  0 \\ 0 & \overline{\lambda} I_r  \end{bmatrix} $. By assumptions, $\Lambda_{d-r} < \overline{\lambda} I_{d-r}$. Thus, there exists $ \epsilon' >0$ such that for all  $\sigma  $,   $f_1(0,\sigma, 0)  < - \epsilon'$ and  $\dv{f}{x} <-\epsilon'$ at $(x,z)=0$, for any $\sigma\in\DD_{d-r}$.
Let us consider now the expansion of $z \mapsto f(0,z)$ up to order $2$ versus $z$. Using $\delta z=(\delta \zeta,\delta \omega)$  and~\eqref{eq:tangentmap}, completed via second order terms derived form  the Backer-Campbell-Hausdorf formula,  we find:
$$
\delta\rho = \begin{bmatrix}\qquad \delta \omega~\overline{\Delta}_r ~\delta\omega^\dag \qquad   &  \delta\omega~ (\overline{\Delta}_r+\delta\zeta)  \\ (\delta\zeta +\overline{\Delta}_r)~ \delta\omega^\dag &
\delta\zeta- \frac{\delta\omega^\dag \delta\omega\overline{\Delta}_r + \overline{\Delta}_r \delta\omega^\dag \delta\omega}{2}   \end{bmatrix} + 0 (\|\delta z\|^3)
.$$ Consequently,
\begin{multline} \label{eq:DLf}
  f(0,\delta z)=f(\brho)+\tr{\left. \nabla f\right|_{\brho} ~\delta \rho} + \tfrac{1}{2} \left. \nabla^2 f\right|_{\brho} (\delta\rho,\delta\rho)   + O(\|\delta\rho\|^3)
  \\
   = f(\brho) -  \tr{(\overline{\lambda} I_{d-r}-\Lambda_{d-r}) \delta\omega ~\overline{\Delta}_r~\delta\omega ^\dag }-  \tfrac{1}{2} \sum_{\mu} \frac{\trr{ \delta \rho Y_\mu} }{\trr{\brho Y_\mu}}
   .
\end{multline}
This  shows that  $\dv{f}{z}$ vanishes at $(0,z)$ and that $\dvv{f}{z}$ is negative definite at $(0,z)$ ($\overline{\lambda} I_{d-r}>\Lambda_{d-r}$) and independent of $\sigma$. All the assumptions necessary for~\eqref{case1:thm:boundary} are fulfilled and we can write:
\begin{multline*}
  \int_{\DD} e^{\Nm f(\rho)} ~\mathbb{P}_0(\rho)~ \mathrm d\rho =
\\  \kappa_0  ~e^{f(\brho)\Nm} \Nm^{-m-n/2-1}
\int_{\sigma\in\DD_{d-r}}
\tfrac{\mathrm d\sigma  }{ \left(\overline{\lambda} - \tr{\Lambda_{d-r}\sigma} \right)^{m+1}}
           + O\Big(e^{f(\brho)\Nm}\Nm^{-m-n/2-2}\Big)
\end{multline*}
where $\kappa_0= \tfrac{\mathbb{P}_0(\brho) J(0,0) m! ~(2\pi)^{n/2}}{\sqrt{\left| \det\left(\left.\dvv{f}{z}\right|_{(0,0)}\right) \right|~ }}$.

Similarly we have;
\begin{multline*}
  \int_{\DD} \tr{\rho A} e^{\Nm f(\rho)}~ \mathbb{P}_0(\rho) ~\mathrm d\rho =
\\  \kappa_0 \tr{A\brho} ~e^{f(\brho)\Nm} \Nm^{-m-n/2-1}
\int_{\sigma\in\DD_{d-r}}
\tfrac{\mathrm d\sigma  }{ \left(\overline{\lambda} - \tr{\Lambda_{d-r}\sigma} \right)^{m+1}}
           + O\Big(e^{f(\brho)\Nm}\Nm^{-m-n/2-2}\Big)
           .
\end{multline*}
Consequently, we have proved that:
$I_A(\Nm)= \tr{\brho A} + O(1/\Nm)$.

Simple computations show that the expansion  of $V_A(\Nm)$  reduces to the expansion of the following integral
$\int_{\DD}\trr{(\rho-\brho) A}  e^{\Nm f(\rho)}~ \mathbb{P}_0(\rho) ~\mathrm d\rho$  based on~\eqref{case2:thm:boundary} with
$g(x,\sigma,z)= J(x\sigma,z) \mathbb{P}_0(x\sigma,z) h(x\sigma,z) $, $h(x\sigma,z)=\trr{(\Upsilon(x\sigma,z)-\brho) A}$ and $z=(\zeta,\omega)$.
Since $\left.\dvv{g}{z}\right|_{(0,\sigma, 0)}= J(0,0) \mathbb{P}_0(\brho) \left.\dvv{h}{z}\right|_{(0,\sigma,0)}$ is independent of $\sigma$, we have using~\eqref{case2:thm:boundary}:
\begin{multline*}
  \int_{\DD} \trr{(\rho-\brho) A} e^{\Nm f(\rho)}~ \mathbb{P}_0(\rho) ~d\rho =
\\  \kappa_0 \tfrac{\tr{-\left.\dvv{h}{z}\right|_{(0,0)} \left(\left.\dvv{f}{z}\right|_{(0,0)}\right)^{-1} }}{2} ~e^{f(\brho)\Nm} \Nm^{-m-n/2-2}
\int_{\sigma\in\DD_{d-r}}
\tfrac{d\sigma  }{ \left(\overline{\lambda} - \tr{\Lambda_{d-r}\sigma} \right)^{m+1}}
\\
           + O\Big(e^{f(\brho)\Nm}\Nm^{-m-n/2-3}\Big)
           .
\end{multline*}
Consequently, we have
$
V_A(\Nm) = \tfrac{\tr{-\left.\dvv{h}{z}\right|_{(0,0)} \left(\left.\dvv{f}{z}\right|_{(0,0)}\right)^{-1} }}{2\Nm} + O(\Nm^{-2}).
$
The fact that the trace in the numerator coincides with $2\tr{A_{\parallel} ~\left(\bF\right)^{-1}\!\!(A_\parallel)}$ results form the following computations.
\begin{itemize}
  \item Formula~\eqref{eq:ProjOrtho} is unitary invariant. In the frame where
  $\brho= \begin{bmatrix} 0 &  0 \\ 0 & \overline{\Delta}_r  \end{bmatrix}$ is diagonal,  the tangent  space  to the manifold at  $\brho$ of rank $r$ Hermitian matrix is given by $\delta\rho$ satisfying~\eqref{eq:tangentmap} with $\delta\xi=0$ and $(\delta\zeta,\delta\omega)$ arbitrary. One can check that~\eqref{eq:ProjOrtho} provides the following block decomposition $\begin{bmatrix} 0 &  A_{0,r} \\ A_{0,r}^\dag \quad&  A_r - \tfrac{\tr{A_r}}{r} I_r \end{bmatrix}$ for $A_\parallel$
   when $A=\begin{bmatrix} A_{0} &  A_{0,r}\\ A_{0,r}^\dag   & A_r \end{bmatrix}$. One can also  check that  $A_\parallel$ belongs to this tangent space and that $\tr{A \delta\rho}= \tr{A_\parallel \delta \rho}$  for any tangent element  $\delta\rho$.

  \item Since $h(0,z)=\trr{(\Upsilon(0,z)-\brho) A}$, we have:
  $$
 \left. \dvv{h}{z}\right|_{(0,0)} (\delta z, \delta z) = 2 \trr{\delta \Upsilon ~A}=  2 \trr{\delta \Upsilon ~A_\parallel}
  $$
  with $\delta \Upsilon=
  \begin{bmatrix}0   &  \delta\omega~ \overline{\Delta}_r \\
  \overline{\Delta}_r~ \delta\omega^\dag &
    \delta\zeta   \end{bmatrix} $ and $\delta z= (\delta\zeta,\delta\omega)$.
    This means that $\left. \dvv{h}{z}\right|_{(0,0)}$ is colinear to the orthogonal projector on  the direction given by $A_\parallel$ in the tangent space to $\brho$.
    This implies that  $\tr{\left.\dvv{h}{z}\right|_{(0,0)} \left(\left.\dvv{f}{z}\right|_{(0,0)}\right)^{-1} }$  corresponds to twice the  value at $A_\parallel$ of the quadratic form attached to the inverse of the  Hessian at $\brho$ of the restriction    of  $f$ to the manifold of rank $r$ Hermitian matrices of  trace one (we use  here lemma~\ref{lem:trace}).

  \item This Hessian is given by~\eqref{eq:FisherInfo} since, for $X= \delta \Upsilon = \begin{bmatrix}0   &  \delta\omega~ \overline{\Delta}_r \\
  \overline{\Delta}_r~ \delta\omega^\dag &
    \delta\zeta   \end{bmatrix} $,   we have:
    \begin{multline*}
     \tr{ X  \left(\overline{\lambda}I-\left.\nabla f\right|_{\brho}\right) X \brho^{+} +  X  \brho^{+} X \left(\overline{\lambda}I-\left.\nabla f\right|_{\brho}\right)}
      \\
      =  2\tr{(\overline{\lambda} I_{d-r}-\Lambda_{d-r}) \delta\omega ~\overline{\Delta}_r~\delta\omega ^\dag }
    \end{multline*}
    because  $\brho^{+}=\begin{bmatrix} 0 &  0 \\ 0 & \overline{\Delta}_r^{-1}  \end{bmatrix} $. We recover from~\eqref{eq:DLf} that
    $f(0,z)= f(\brho) - \tfrac{1}{2} \tr{ X~ \bF(X)} $, i.e., that $\bF$ is indeed  the Hessian at $\brho$ of the restriction of $f$ to rank-r Hermitian matrices of  trace one.

\end{itemize}

   \end{proof}

\section{Concluding remark}

When maximum likelihood estimation provides a quantum state of reduced rank, we have provided, based on asymptotic expansions of specific multidimensional  Laplace integrals, an estimate of  the Bayesian mean and variance for any observable.  We guess that similar asymptotic expansions   could be of some interest  for quantum compress sensing~\cite{GrossLFBE2010PRL} when the dimension of underlying Hilbert space is large and the rank is small.

\bibliographystyle{plain}

\end{document}